\documentclass[journal,twcolumn]{IEEEtran}

\usepackage{color}
\usepackage{graphicx}
\usepackage{amsfonts}
\usepackage{amsmath}
\usepackage{amssymb}
\usepackage{psfrag}
\usepackage{cite}
\usepackage{algorithmic}
\usepackage{amsthm}
\usepackage{epsfig}
\usepackage{epstopdf}
\usepackage{subfig}
\usepackage{float}
\usepackage{stfloats}
\usepackage{verbatim}

\usepackage[longend,ruled]{algorithm2e}
\usepackage{booktabs}
\usepackage[numbers]{natbib}

\newcommand{\mat}{\boldsymbol}

\newtheorem{proposition}{\it Proposition}

\begin{document}

\title{On Maximizing Information Reliability in Wireless Powered Cooperative Networks}

\author{\IEEEauthorblockN{M. Majid Butt\IEEEauthorrefmark{1},~Galymzhan~Nauryzbayev\IEEEauthorrefmark{2},~Nicola~Marchetti\IEEEauthorrefmark{3}}\\
\IEEEauthorblockA{\IEEEauthorrefmark{1}Nokia Bell Labs, Paris-Saclay, France,
Email: majid.butt@nokia-bell-labs.com}\\
\IEEEauthorblockA{\IEEEauthorrefmark{2}Nazarbayev University, Nur-Sultan, 010000, Astana, Kazakhstan\\
Email: galymzhan.nauryzbayev@nu.edu.kz}\\
\IEEEauthorblockA{\IEEEauthorrefmark{3}Trinity College, University of Dublin, Ireland,
Email: nicola.marchetti@tcd.ie}
	\thanks{This work was supported in part by Science Foundation Ireland under European Regional Development Fund under Grant 13/RC/2077 and the Nazarbayev University Faculty Development Competitive Research Program under Grant SEDS2020014.}
}

\maketitle

\begin{abstract}
Unpredictable nature of fading channels and difficulty in tracking channel state information pose major challenge in wireless energy harvesting communication system design. In this work, we address relay selection problem for wireless powered communication networks, where the relays harvest energy from the source radio frequency signals. A single source-destination pair is considered without a direct link.
The connecting relay nodes are equipped with storage batteries of infinite size.
We assume that the channel state information (CSI) on the source-relay link is available at the relay nodes. Depending on the availability of the CSI on the relay-destination link at the relay node, we propose two relay selection schemes and evaluate their outage probability. Availability of the CSI at the relay node on the relay-destination link considerably improves the performance due to additional flexibility in the relay selection mechanism.
Due to absence of CSI throughout the network at the time of transmission path selection, the analysis of the problem is not tractable. Therefore, we relax our assumptions on availability of CSI and closed-form expressions of the outage probability as a function of the amount of the available harvested energy are derived for both CSI availability cases.
Finally, we numerically quantify the performance for the proposed schemes and compare the outage probability for fixed and equal number of wireless powered forwarding relays.
\end{abstract}
\begin{IEEEkeywords}
Relay selection, RF energy harvesting, outage probability, SWIPT, wireless powered communication networks.
\end{IEEEkeywords}

\section{Introduction}
Wireless powered communication networks (WPCNs) are one of the promising technologies to achieve sustainable wireless networks, where the communicating nodes are powered by radio frequency (RF) signals. The simultaneous wireless information and power transfer (SWIPT) concept has been investigated extensively to realize WPCNs, e.g., authors in \cite{He_commag_2015,AHMED_compnet} survey recent works in energy harvesting communication domain. In SWIPT, energy and information are transferred from the same RF signal by using either time sharing or power splitting protocol \cite{Zhang_IEEECOM:2013,Ali_IEEEWS:2013, Gao_transcognitive_2019}. The time sharing protocol allocates dedicated time for energy harvesting and information transfer, while power splitting extracts energy and information from the same RF signal.

Relays are used in wireless networks to extend the coverage and increase the information reliability. Relay selection problem using amplify-and-forward (AF) and decode-and-forward (DF) techniques in WPCNs has been addressed in literature quite extensively, e.g., \cite{Krikidis_COML:2012,Michalopoulos:JSAC2015}. Similar to the traditional relay selection schemes with fixed power supply, battery status information (BSI) enabled relay selection schemes are proposed to improve the system performance of multi-relay wireless powered cooperative networks by exploiting both channel state information (CSI) and BSI to make the selection decision \cite{Gu_add1,Liu_add1}. In such systems, the relaying nodes, with BSI indicating the amount of stored energy above a predetermined power threshold required for successful information transmission, will first create a subset and then will send their CSI back to the source node. By the next step, the 'best' relay from the subset will be chosen to forward the source information, while the remaining relays will harvest energy from the source information signal.

In literature, there are three main relaying architectures for SWIPT systems, namely, ideal relaying receiver (IRR), power-splitting relaying (PSR) and time-switching relaying (TSR) protocols \cite{Ali_IEEEWS:2013,Krikidis_COML:2012,Galym_pimrc,Michalopoulos:JSAC2015,Zhong:TCOM2014,Galym_globecom,Galym_vtc2018fall,Nasir_TCOM:2015}. A recent work in \cite{Michalopoulos:JSAC2015} discusses an RF based cooperative network, where the relays are used for transmitting information to a designated receiver and for transmitting energy to an associated ambient RF energy harvester. For the case where the number of relays is more than two, two relay selection methods are developed and the trade-off between outage probability and average energy transfer is discussed. The authors in \cite{Zhong:TCOM2014} analyze the performance of a network consisting of a single source, single relay and single destination. The throughput of the TSR protocol for both AF and DF relaying schemes for the same model is also investigated in \cite{Nasir_TCOM:2015}. A similar system model is considered in \cite{Zhu_TCOM_2015}, where the relay is equipped with multiple antennas and the outage probability and ergodic capacity of the system are studied. In \cite{Galym_globecom,Galym_ACCESS}, the authors investigate the ergodic capacity and outage probability over $\alpha-\mu$ fading channels of an AF-based network consisting of the source, relay and destination nodes for the IRR, PSR and TSR protocols. The outage probability analysis in two-hop DF and AF relaying systems over log-normal fading channels is provided in \cite{Khaled2,Khaled1}, respectively.  Another work on the outage is presented in \cite{HYu_add1}, where the authors study multi-relay wireless powered cooperative systems over Nakagami-$m$ fading channels.

In \cite{Galym_vtc2018fall}, the authors analyze the outage performance of the wireless powered full-duplex (FD) AF and DF relaying networks in $\alpha-\mu$ environment. Furthermore, the outage in two-way (TW) FD relaying networks with multiple pairs of users is investigated in \cite{Xia_add1}. A similar model with multiple source-destination pairs communicating through a single energy harvesting relay is considered in \cite{Ding_TWCOM:2013} and the system effect of the harvested energy distribution among the users is investigated. In \cite{Senanayake_add1}, the authors analyze outage performance obtainable under a decentralized relay selection strategy in multi-user multihop DF based relaying networks over Nakagami-$m$ fading channels. Multihop relaying for a cognitive radio network is investigated in \cite{8233153}, where the authors aim at minimizing end-to-end outage probability for a secondary user under the energy causality and primary user cooperation rate constraints. The system outage probability is investigated in \cite{Li_add1}, where the authors propose optimal and sub-optimal joint relay-antenna selection schemes for TW AF relaying networks. Moreover, the authors in \cite{Wang_add3} introduce a general relay selection strategy for the PSR-enabled TW FD relaying network, where average sum capacity and outage probability are studied. Other performance metrics such as energy efficiency and security issues in wireless power transfer enabled relaying networks are studied in \cite{Chang_TVT,Wang_add1,Wang_add2,Mohammadi_add1}. For instance, the authors in \cite{Wang_add1,Wang_add2} propose relay selection schemes to improve secrecy outage probability in cooperative DF and AF relaying networks under the presence of eavesdroppers. In \cite{Mohammadi_add1}, joint relay selection and power allocation scheme is proposed for large-scale multiple-input multiple-output AF-based relaying systems with passive eavesdroppers. The authors provide closed-form expressions for ergodic secrecy rate and secrecy outage probability over Rayleigh fading channels.

\subsection{Motivation and Contributions}
In this work, we aim to minimize the outage probability for a cooperative system, comprising of a single source, multiple energy harvesting (EH) relays and a single destination.
A single relay is selected to forward the source signal to the destination.
We assume half duplex relay communication such that the relays receive the source signal in a time slot $t$ and the selected relay forwards it to the destination in time slot $t+1$.
The channels on both the source-relay and the relay-destination links are mutually independent, and are independently and identically distributed (i.i.d.). Due to mutual independence and i.i.d. channel assumption on both links, relay selection poses new challenges as the relay selected to receive data from the source in time slot $t$ will have a completely independent (and unknown) channel realization in time slot $t+1$ for transmission on the relay-destination link.

For a similar setting, the work in \cite{Butt_globalsip,majid:ietsp16} assumes that the CSI is not available on the source-relay link at the relay node. In contrast, we assume availability of the CSI on the source-relay link at the relay node throughout this work. When the CSI is available on the source-relay link, the relay selection exploits the CSI to decide which relays are dedicated for data/energy transfer \cite{Majid_globecom}. Then, conditioned on the availability of transmit CSI (CSIT) at relay on the relay-destination link, we propose novel relay selection schemes.
Furthermore, we derive closed-form analytical expressions of the outage probability which are unified in terms of the amount of the harvested power.
The consideration of the availability of both transmit and receive CSI at the relay node requires a different relay selection approach as compared to relay selection in \cite{majid:ietsp16}.
We formulate the outage minimization problem and evaluate the performance of the proposed relay selection schemes numerically. Then, we compare our schemes with the existing schemes available in literature and show their superiority in terms of outage performance.

The remainder of this paper is structured as follows. Section \ref{sect:system} introduces the system model and the fundamentals for the problem. The novel schemes are proposed in Section \ref{sect:schemes} and the performance is evaluated in Section \ref{sect:results}. Finally, the paper outlines the main concluding remarks in Section \ref{sect:conclusion}.

\section{System Model}
\label{sect:system}
Consider a two-hop DF wireless communication system where a source node $(S)$ transmits information to a destination node $(D)$ by means of one relaying node $(L_{i*})$ selected from $N$ available energy-limited relay nodes (a relay node $i$ is denoted by $L_i$). Moreover, due to the apartness between $S$ and $D$, we assume that there is no direct link between them. We consider a broadcast channel between the source and the relays and assume that all nodes are equipped with a single antenna. The source-to-relay ($S\to L_i$) and relay-to-destination ($L_i\to D$) links, indicated by $h_{si}$ and $h_{id}$, are subject to quasi-static i.i.d. Rayleigh fading. The $S \to L_i$ and $L_i\to D$ distances are denoted by $d_{si}$ and $d_{id}$, respectively; the corresponding path-loss exponent $\alpha$ is chosen to be identical for all the links.

The CSIT is not available at the source and, therefore, the source transmits with a fixed power $P_s$. The received signal $ y_i(t)$ at the relay $L_i$ is expressed as,
\begin{equation}
y_i(t)= \frac{1}{\sqrt{d_{si}^{\alpha}}}\sqrt{P_s}h_{si} x(t)+n(t), i\in\{1,\ldots,N\},
\label{3}
\end{equation}
where $P_s$, $ x(t) $ and $n(t)\sim \mathcal{CN}(0,\sigma^2)$ denote the source transmit power, the normalized information signal and the Gaussian noise with zero mean and variance $\sigma^2$, respectively.

\begin{figure}
\centering
  	\includegraphics[width=3.5in]{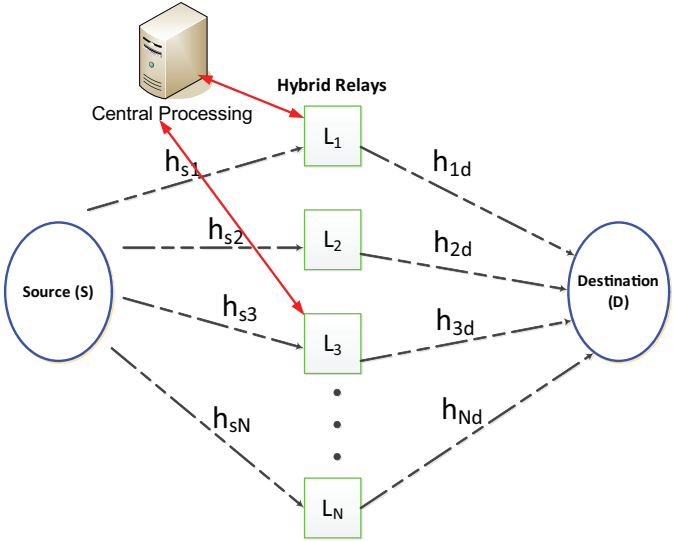}
   \caption{Schematic diagram for the system model. The centralized controller collects relay stored energy information and CSI at relay on the $S\to L_i$ link for all the relays, and then makes relay selection. Note that all the links from the relays to the centralized server are not shown to make the diagram clear.}
	\label{fig:system}
\end{figure}

We make the following assumptions regarding the cooperative system.
\begin{itemize}
\item We assume that the fading coefficients remain constant for the duration of a time slot, but vary independently from one slot to another. The schematic diagram for the system model is shown in Fig. \ref{fig:system} with a centralized controller, which collects stored energy and channel state information from all the relays and makes relay selection in every time slot.
\item We use time sharing protocol for SWIPT at relay nodes.
\item The transmissions are interference free and interference cannot be exploited for energy harvesting at relays.
\item We consider a half duplex communication system. The selected relay node $L_{i^*}$ decodes information from the source in a time slot of duration $T$ and forwards it to the destination in the next time slot. Hence, node $L_{i^*}$ is not available for information (or energy) reception in time slot $t+1$. However, all other nodes are available for information/energy reception from the source signals, thereby mimicking a FD communication system \cite{Ikhlef_TVT:2012}.
\item A single relay is selected for forwarding information to the destination. It is well known that transmission from multiple relays provides spatial diversity and improves data reliability at the cost of increased complexity. We focus on single relay selection schemes to reduce the complexity of the system, but all the proposed schemes can be extended to multiple relay transmission schemes in a straight forward manner.
\item The circuit energy consumed in energy harvesting and information decoding at relay is negligible. However, if circuit energy consumption is not negligible but same for all relay nodes, it will not have any effect on the design of proposed relaying schemes as all the relay nodes behave symmetrically. Therefore, circuit power consumption can be neglected in this problem without affecting problem and solution design.
\item We assume that the relay has no external power supply (i.e., powered by energy-harvesting from source-transmitted signal) while being deployed with a battery of infinite capacity. It is also assumed that there is negligible leakage within the time period of interest \cite{Yuan:ICC2014}.
\end{itemize}
The rate $R_{si}$ provided on the $S\to L_i$ link in a time slot $t$ is given by
\begin{equation}
R_{si}= \frac{1}{2}\log_2 \left( 1+|h_{si}|^{2}\frac{P_s}{\sigma^2} \right),
\label{eqn:rate_equation1}
\end{equation}
where the factor $\frac{1}{2}$ shows that the $S$-to-$D$ transmission requires two time slots. For a relay transmit power $P_{r}$, the rate $R_{id}$ provided by the $L_i\to D$ link is given by,
\begin{equation}
R_{id}= \frac{1}{2}\log_2 \left(1+|h_{id}|^{2}\frac{P_{r}}{\sigma^2}\right)~,
\label{eqn:rate_equation2}
\end{equation}
where $h_{id}$ denotes the channel coefficient between the relay node $L_i$ and the destination D.

All the nodes selected for energy transfer harvest energy from the source signal. Assuming $T=1$ and linear energy harvesting model, without loss of generality, the energy harvested by a relay node is given by
\begin{equation}
E_i^h=\eta P_s\left | h_{si} \right |^2,
\label{4}
\end{equation}
where $0\leq\eta\leq 1$ is the energy harvesting efficiency. In recent literature, non-linear energy harvesting models have been used for various studies, e.g., \cite{Boshkovska,Alevizos}, which account for non-linear effects in RF energy harvesting in practical systems. For simplicity, we assume a linear model here to focus more on the impact of CSI assumptions on system outage performance.

\subsection{Problem Formulation}
For a traditional grid powered DF relaying strategy, the outage probability $P_{\rm out}$ that a rate $R$ is not supported by the system is given by
\begin{align}
\label{outage}
P_{\rm out}&=\Pr\left(\min\left(R_{si^*},R_{i^*d}\right)<R\right)\\
&={\Pr} \left(\min\left(\frac{1}{2}\log_2 \left(1+ | h_{si^*}|^{2}\frac{P_s}{\sigma^2}\right), \right.\right. \nonumber\\
&\hspace{1.95cm}\left. \left. \frac{1}{2}\log_2 \left(1 + | h_{i^*d}|^{2}\frac{P_{r}}{\sigma^2}\right)\right) < R \right),
\label{eqn:DF}
\end{align}
where $L_{i^*}$ is the selected node for information relaying.
However, when the relays are powered by EH, unavailability of the harvested energy is an additional source of outage.

For a wireless powered relay network, the outage probability is given by,
\begin{equation}
P_{\rm out}= 1-\zeta_s\zeta_p\zeta_r,
\label{eqn:outage_EH}
\end{equation}
where $\zeta_s$ and $\zeta_r$ denote the success probability on the $S\to L_i$ and $L_i\to D$ links, respectively while $\zeta_p$ is the probability that the selected relay can support a transmit power $P_r$ to forward the information to the destination. $\zeta_p$ depends on the relay selection scheme and energy harvesting efficiency $\eta$. If there is always sufficient energy available for decoding, $\zeta_p\to 1$ and the EH system behaves as a grid powered system.

Our objective is to find a relay selection policy $\pi$ that minimizes the outage probability for the EH cooperative system with a fixed number of relays.
The optimization problem is formulated as:
\begin{eqnarray}
&\min_{\pi}&~ P_{\rm out}=\Pr\Big(\min(R_{si^*},R_{i^*d})<R\Big)\\
&{\rm s.t}.&~
\begin{cases}
\mathcal{C}_1:N=a,& a\in \mathbb{N}\\
\mathcal{C}_2:P_s = b\\
\mathcal{C}_3:R\geq 0\\
\mathcal{C}_4:E_{i^*}^{\rm st}(t)>E_r(t)
\end{cases}
\label{eqn:cons}
\end{eqnarray}
where $\mathcal{C}_1$ and $\mathcal{C}_2$ are constants representing a fixed number of relay nodes in the system and the fixed source power, respectively. $\mathcal{C}_3$ is the rate condition. $\mathcal{C}_4$ is neutrality constraint which implies that stored energy $E_{i^*}^{\rm st}$, for the node selected for forwarding $L_{i^*}$, must be greater than the energy $E_r$ required for transmission. We intend to find a relay selection policy $\pi$ which follows the constraints in ($\ref{eqn:cons}$) and minimizes network failure (outage probability).

A closed-form solution for the problem is difficult to achieve due to involvement of multiple relays and multiple energy queues at the relay nodes, which depend on stochastic fading channels. The energy queue states are mutually coupled for the outage analysis, making computation of $\zeta_p$ in (\ref{eqn:outage_EH}) difficult, and the analysis is not tractable for a large number of relays. Therefore, we propose heuristic relay selection schemes and evaluate the outage performance numerically. To give insight, we provide outage performance analysis for a certain relay $i$ depending on the number of energy queues.

\section{Relay Selection Schemes}
\label{sect:schemes}
To avoid the outage event, there must be at least a single relay available with sufficient energy to transfer data to the destination in time slot $t+1$. To minimize the outage probability, the relay selection scheme should aim at maximizing energy harvesting of the relays and minimizing the energy expenditure on the $L_i\to D$ link, thereby maximizing the network lifetime with minimum failure.

A simplified approach to model such a system is to assume the length of fading blocks long enough such that the channels on the $S\to L_i$ and $L_i \to D$ links remain constant for both reception and forwarding phases at the relay \cite{Yaming,Ju_IEEEWCOM:2014}. This has the advantage that both receive and transmit channels are known at the time of relay selection. In contrast, we assume that the relay reception and transmission occurs in two consecutive time slots.
Due to mutually independent channels on the $S\to L_i$ and $L_i\to D$ links, and the fact that the CSI for the $L_i \to D$ channel is not available when a relay is selected for forwarding at time $t$, relay selection becomes challenging.

The relay selection scheme is dictated by the availability of the CSI at relay on the both $S\to L_i$ and $L_i \to D$ links. Throughout this work, we assume that the CSI is available at the relay on the $S\to L_i$ link\footnote{The CSI estimation at the relay can be performed by pilot/data aided techniques.}.

Regarding the availability of the CSI at relay node on the $L_i \to D$ link, we consider the following two scenarios:
\begin{enumerate}
  \item The CSIT is not available on the $L_i \to D$ link.
  \item The CSIT is available on the $L_i \to D$ link.
\end{enumerate}

\subsection{No CSIT on the $L_i \to D$ Link}
When the CSIT is not available on the $L_i \to D$ link, no power allocation can be performed. Therefore, the selected relay transmits with a fixed power $P_{r}$ at time $t+1$.

At time $t$, the relay selection is performed. As the CSIT on $L_i \to D$ link is not available, relay selection is solely based on the available information at time $t$. As a single relay is selected at time $t$, this scheme is called single relay selection with No CSIT (SRS-NCSI).

A relay is selected for decoding information such that,
\begin{eqnarray}
\label{eqn:SRS_selection}
i^*=\arg \min_i R_{si},
\end{eqnarray}
where (\ref{eqn:SRS_selection}) is evaluated for a relay $i$ only if,
\begin{equation}\label{eqn:SRS}
   I(R_{si}>R)\times I\left(\frac{E_i^{\rm st}}{T}>P_r\right)=1
\end{equation}
such that,
\begin{equation}
I(R_{si}>R)=\begin{cases}0 & R_{si}<R\\
1 & R_{si}\geq R
\end{cases},
\label{eqn:SRS1}
\end{equation}
and
\begin{equation}
I\left(\frac{E_i^{\rm st}}{T}>P_r\right)=\begin{cases}0 & \frac{E_i^{\rm st}}{T}<P_r\\
1&\frac{E_i^{\rm st}}{T} \geq P_r
\end{cases}.
\label{eqn:SRS2}
\end{equation}
$E_i^{\rm st}$ denotes the stored energy for relay $L_i$.
The indicator functions $I(R_{si}>R)$ and $I\big(\frac{E_i^{\rm st}}{T}>P_r\big)$ in (\ref{eqn:SRS}) ensure that a selected node can decode the signal from the source and has energy to transmit with a fixed power $P_{r}$ in time slot $t+1$. Equation (\ref{eqn:SRS_selection}) selects the node with the minimum $R_{si}$ for information decoding out of the nodes which satisfy (\ref{eqn:SRS}).
The rationale behind the selection of the node with minimum $R_{si}$ is to provide relatively 'average' $S\to L_i$ channel for information decoding at relay as information decoding is already ensured by the condition $I(R_{si}>R)$. This implies that good $S\to L_i$ channels (which satisfy (\ref{eqn:SRS})) can be better utilized for energy harvesting as decoding on the best channel does not improve the outage performance as long as (\ref{eqn:SRS1}) is satisfied. If $R_{i^*d}<R$ for the selected relay or no relay satisfies (\ref{eqn:SRS}), an outage event occurs.

All the relays with $i\ne i^*$ harvest energy from the source signal such that,
\begin{equation}
\label{energy_NCSI1}
E_i^{\rm st}(t+1) = E_i^h(t)+ E_i^{\rm st}(t), \quad i\ne i^*~.
\end{equation}
The selected relay node $L_{i^*}$ is not a candidate for selection for both decoding and harvesting from the source signal in time slot $t+1$, which implies that $E_{i^*}^{\rm st}(t+1)=E_{i^*}^{\rm st}(t)$ and the stored energy for node $L_{i^*}$ after making a transmission at time $t+1$ is given by,
\begin{equation}
\label{energy_NCSI2}
E_{i^*}^{\rm st}(t+2)= E_{i^*}^{\rm st}(t+1)-P_{r}T~.
\end{equation}
It is clear from \eqref{eqn:DF} that $P_{\rm out}$ is determined by the rate provided by the 'bottleneck' link. However, the outage probability in WPCN is also characterized by the amount of energy harvested by the relay nodes. The harvested energy is a function of the source power, channel distribution and the energy harvesting efficiency $\eta$. When $\eta$ is large, very small number of relay nodes provide enough stored energy such that there is always a node available with enough energy to transmit on the $L_i \to D$  link and the outage probability in (\ref{eqn:outage_EH}) converges to (\ref{eqn:DF}). However, when $\eta$ or $N$ is small, (\ref{eqn:DF}) is only a lower bound on $P_{\rm out}$.

It is worthwhile pointing out that the results on outage performance associated with \eqref{eqn:SRS_selection} - \eqref{eqn:SRS2} will take into account the number of relays $N$ and randomness of each channel realization of all possible $S\to L_i$ and $L_i \to D$ links to calculate the amount of harvested energy. However, it is intractable to consider these features in the outage performance analysis since it is not feasible to capture a dynamic behavior of the EH and communication activities in one generalized analytical framework. Hence, our analysis will be performed for a relay $i$ at certain time slot (which can be regarded as a static process with initial inputs on the stored power available for successful end-to-end communication) with various levels of the harvested power available for information transmission.

\begin{proposition}
	\label{prop1}
	{\rm With respect to the defined relay selection strategy, the outage probability given in \eqref{eqn:outage_EH} can be rewritten in its closed-form as in \eqref{Pout_strategy1}, shown at the top of the current page. The generalized expression of the outage probability provided here considers not only the number of time slots allocated for EH purposes ($k$) but also how many times a relay $i$ was chosen for data transmission ($n$).}
\end{proposition}
\begin{proof}
	See Appendix \ref{sect:Appen1}.
\end{proof}
\begin{figure*}[!t]
	\begin{align}
	\label{Pout_strategy1}
	P_{\rm out} = \begin{cases}
	\begin{array}{ll}
	1, & \text{for}~k=1,\\
	1 - \exp\left( -\frac{\lambda v \sigma^2}{P_r} \right)
	\exp\left( - \frac{\lambda v\sigma^2}{P_s}\right) \exp\left(-\lambda\frac{(n+1)P_r }{\eta P_s}\right), & \text{for}~k=2,\\
	1 - \exp\left( -\frac{\lambda v \sigma^2}{P_r} \right)
	\exp\left( - \frac{\lambda v\sigma^2}{P_s}\right) \left(1 - \frac{1 - \left(\lambda \frac{(n+1)P_r }{\eta P_s} + 1 \right)\exp\left(-\lambda \frac{(n+1)P_r }{\eta P_s}\right)}{\lambda}\right), & \text{for}~k=3,\\
	1 - \exp\left( -\frac{\lambda v \sigma^2}{P_r} \right)
	\exp\left( - \frac{\lambda v\sigma^2}{P_s}\right) \left(1 - \frac{\gamma_{inc}\left(k-1,\lambda \frac{(n+1)P_r }{\eta P_s}\right)}{\Gamma(k-1)}\right), & \text{for}~k>3.
	\end{array}
	\end{cases}
	\end{align}
	\hrulefill
\end{figure*}


\subsection{CSIT Available on the $L_i \to D$  Link}
\label{sect:scheme_ACSI}
In the case when the CSIT is available at the relay on the $L_i \to D$ link, the relay can benefit from this information through power allocation. CSI can be made available at relay node either by using data aided estimation or making use of explicit training sequences. As we show later in this section, this helps improve system outage performance and overhead in channel estimation is justified. This feature is commonly used in 4G and 5G systems, especially when there are no latency constraints.

The required transmit power to make a successful transmission for a relay $L_i$ is computed from (\ref{eqn:rate_equation2}), and given by,
\begin{equation}
P_{id} = \frac{(2^{2R}-1)\sigma^2}{|h_{id}|^{2}}.
\label{eqn:power}
\end{equation}
This scenario provides more flexibility for relay selection. However, the CSIT on the $L_i \to D$ link is available only at time $t+1$ due to i.i.d. channel assumption and the relay is selected at time $t$.

The main challenges in relay selection are:
\begin{enumerate}
  \item If the relay selection is made based on the channel quality on the $S\to L_i$ link, the selected relay node may not have enough energy to transmit on the $L_i \to D$ link at time $t+1$. At the same time, the use of channel quality on the $L_i \to D$ channel will not be optimal as the selected relay $L_{i^*}$ may not necessarily have the best channel at time $t+1$.
  \item If the relay selection is based on the stored energy maximization, the availability of the CSI on $S\to L_i$ and $L_i \to D$ links at the relay node is not exploited.
\end{enumerate}
To take the advantage of CSI availability at different times, we propose a two step relay selection algorithm.\\
\textbf{Phase I}: In the first phase, a subset $\Gamma$ of (maximum) $M$ relays\footnote{To avoid confusing it with multiple relay forwarding, please note that only one relay will be selected for forwarding after phase 2 of the scheme.} is selected out of $N$ relays for decoding information. As the CSIT on the $L_i \to D$ link is not known at time $t$, more than one relay decode information to provide multiuser diversity for the transmission on the $L_i \to D$ link. It is worth noting that selecting a single relay in phase 1 makes available CSIT on the $L_i \to D$ link at time $t+1$ useless as all other relays cannot be used for forwarding in phase 2 of the scheme. Due to multiple relay selection in first phase, this scheme is termed as Multiple Relay Selection with available CSI (MRS-ACSI).

The selection for the forwarding set $\Gamma$ is made such that,
\begin{eqnarray}
\label{eqn:csi_gamma2}
\Gamma^{K\times 1}=\{i:\gamma_i\leq \gamma_K\}
\end{eqnarray}
where (\ref{eqn:csi_gamma2}) is evaluated for a relay $i$ only if,
\begin{equation}
  I(R_{si}>R)=1.
\end{equation}
$\gamma_K$ denotes the fading channel with $K^{\rm th}$ smallest amplitude, selected out of $U$ nodes satisfying $I(R_{si}>R)$. Cardinality $K$ of $\Gamma$ is limited by $\min(M,U)$, where $M\leq N$ is a system parameter for the scheme. Equation (\ref{eqn:csi_gamma2}) states that $K$ relays with the smallest fading channels are dedicated for decoding information.
This metric chooses the relay nodes with the weakest channels, but capable of decoding the information. This implies that the rest of the $N-K$ relays harvest energy from the source signals on good channels and the stored energy for the nodes increases at a faster rate.\\
\textbf{Phase II:} In the second phase of the relay selection algorithm, the forwarding relay from the set $\Gamma$ at time $t+1$ is selected such that,
\begin{eqnarray}
\label{eqn:MRS_Ph2}
i^*=\arg \min_{i\in \Gamma} P_{id}
\end{eqnarray}
where (\ref{eqn:MRS_Ph2}) is evaluated for a relay $i$ only if,
\begin{equation}
  I\Big(\frac{E_i^{\rm st}}{T}>P_{id}\Big)=1
\label{eqn:MRS_cons}
\end{equation}
such that,
\begin{equation}
I\Big(\frac{E_i^{\rm st}}{T}>P_{id}\Big)=\begin{cases}
1& \frac{E_i^{\rm st}}{T}\geq P_{id}\\
0&\frac{E_i^{\rm st}}{T}<P_{id}
\end{cases}.
\label{eqn:MRS_ph2cond}
\end{equation}
The scheme selects the relay with the best transmit channel out of the relays, which have enough stored energy for transmission as in constraint (\ref{eqn:MRS_cons}). This ensures transmission with minimum expenditure and is the optimal decision for the relays in $\Gamma$. Note that $P_{id}$ is calculated individually for every relay $L_i$ via (\ref{eqn:power}). If the cardinality of $\Gamma$ set is zero or no relay in the set satisfies (\ref{eqn:MRS_cons}), outage occurs.
The pseudocode for the second phase of the relay selection algorithm is presented in Algorithm 1.

\begin{algorithm}
\label{algorithm}
\caption{Routine for Relay Selection}
\KwIn{$\mat{h}_{d},\mat{E}^{\rm st}$}
$\mat{E}^{\rm st}$= Vector of stored energies for the relays $i\in \Gamma$\;
$\mat{h}_d=$ Vector of fast fading for the relays $i\in \Gamma$\;
$K$ = Size of vector $\Gamma$\;
\tcc{Initialize outage flag.}
$P_{\rm out}=0$\;
\tcc{Compute the required power vector $\mat{P}_{id}$ to transmit rate $R$ for $i\in \Gamma$.}
\For{i=1 \KwTo K}{
${P}_{id}=\frac{(2^{2R}-1)\sigma^2}{|h_{id}|^{2}}$\;
}
$\Gamma_s$= Sort $\Gamma$ in increasing order w.r.t. $\mat{P}_{id}$\;
$i^*=0$\;
\For{i=1 \KwTo K}{
temp =$\Gamma_s(i)$\;
 \If{$\frac{\mat{E}^{\rm st}(\rm temp)}{T}\geq \mat{P}_{id}(\rm temp)$}{
 break\;
 $i^*= i$\;
}
}
\If {$i^*==0$}{
$P_{\rm out}$ = 1\;
}
\Return{$i^*,P_{\rm out}$;}

\end{algorithm}



After transmission, the stored energy for the node $L_{i^*}$ is updated as,
\begin{eqnarray}
E_{i^*}^{\rm st}(t+2)= E_{i*}^{\rm st}(t+1)-P_{i^*d}(t+1)T~.
\end{eqnarray}
The rest of the nodes harvest and store energy depending on the received signal strength from the source such that
\begin{eqnarray}
E_j^{\rm st}(t+1)=
\begin{cases}
E_j^h(t)+ E_j^{\rm st}(t),&   j \notin \Gamma\\
E_j^{\rm st}(t),&   j \in \Gamma,j \ne i^*
\end{cases},
\end{eqnarray}
with the nodes $j\in \Gamma,j\ne i$ not able to harvest energy as they were reserved for decoding.

\begin{figure*}[!b]
	\hrulefill
	\begin{align}
	\label{Pout_strategy2}
	P_{\rm out} = \begin{cases}
	\begin{array}{ll}
	1, & \text{for}~k=1,\\
	1 - 2
	\sqrt{ \frac{\lambda (n+1)v \sigma^2}{\eta P_s}}
	\exp\left(- \frac{\lambda v \sigma^2}{P_s}\right)
	K_1\left( 2\sqrt{ \frac{\lambda (n+1)v \sigma^2}{\eta P_s}} \right), & \text{for}~k=2, \\
	1 - \frac{1}{\Gamma(k-1)}\left(\frac{\lambda (n+1)v \sigma^2}{\eta P_s}\right)^{k-1} \exp\left(- \frac{\lambda v \sigma^2}{P_s}\right) H^{2,0}_{0,2}\left(\frac{\lambda (n+1)v \sigma^2}{\eta P_s} \left|
	\begin{array}{c}
	-\\
	(0,1),~(-(k-1),1)
	\end{array}
	\right.\right), & \text{for}~k\ge 3.
	\end{array}
	\end{cases}
	\end{align}
\end{figure*}
\begin{proposition}
	\label{prop2}
	{\rm According to the chosen relay selection strategy, the outage probability given in \eqref{eqn:outage_EH}, when CSIT is available at the relay, can be expressed in its closed-form as in \eqref{Pout_strategy2}, shown at the bottom of the next page. The generalized solution for the outage probability provides an insight into how the number of time slots allocated for EH purposes ($k$) as well as the number of time slots devoted for data transmission ($n$) will affect the outage performance of the communication system under consideration.}
\end{proposition}
\begin{proof}
	See Appendix \ref{sect:Appen2}.
\end{proof}

There is a tradeoff involved with the selection of parameter $M$ for a fixed $N$. If $M$ is large, there is greater chance of finding a good channel for transmission on the $L_i \to D$ link, but fewer relays are available for EH and the relay system becomes power limited. On the contrary, if $M$ is too small, less multiuser diversity is exploited on the $L_i \to D$  link, but more relays harvest energy. Thus, for the proposed scheme, it is important to optimize $M$ for a given $N$ and $\eta$.

Given that we have a multiple relay selection (MRS-ACSI) policy $\pi(M,N)$, the parameter optimization problem is formulated by
\begin{eqnarray}
M^*(N,\eta,R) &=& \arg\min_{\pi(M,N),~0<M\leq N} P_{\rm out},
\end{eqnarray}
with the same constraints as in (\ref{eqn:cons}). We determine the optimal $M$ for the proposed scheme numerically in Section \ref{sect:results}.

\section{Numerical Results}
\label{sect:results}

We numerically evaluate the performance of the proposed schemes in this section. A Rayleigh fading channel with mean one is considered on the $S\to L_i$ and $L_i \to D$ links. Time slot $T$ is assumed to be one while noise variance $\sigma^2=1$. $P_s$ is fixed to $10$ dBW throughout.
\begin{figure}[!t]
	\centering
	\includegraphics[width=3.5in]{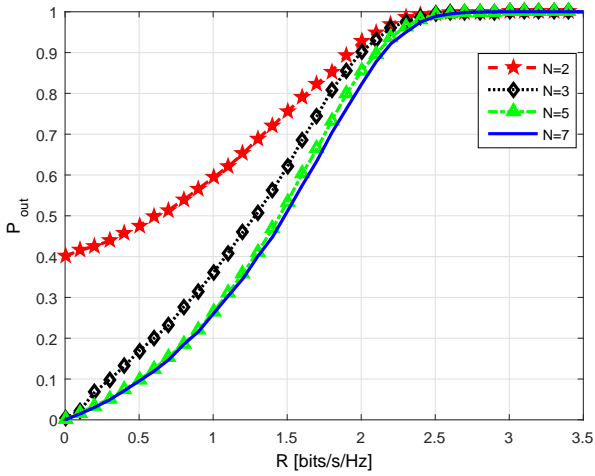}
	\caption{Outage probability for the SRS scheme for different $N$ and $\eta=0.7$.}
	\label{fig:SRS}
\end{figure}

\begin{figure}[!t]
	\centering
	\subfloat[Relay selection probability when $n=0$.]{
		\label{subfig:NCSI_prob_n0}
		\includegraphics[width=3.5in]{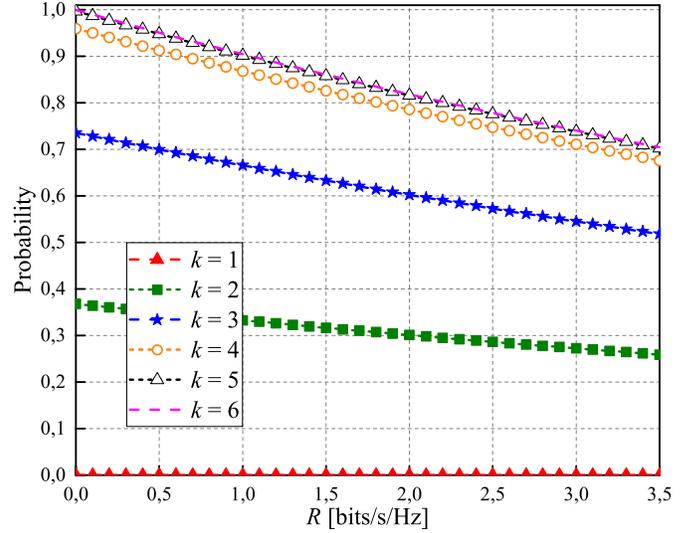}}\\
	\subfloat[Relay selection probability when $n=1$.]{
		\label{subfig:NCSI_prob_n1}
		\includegraphics[width=3.5in]{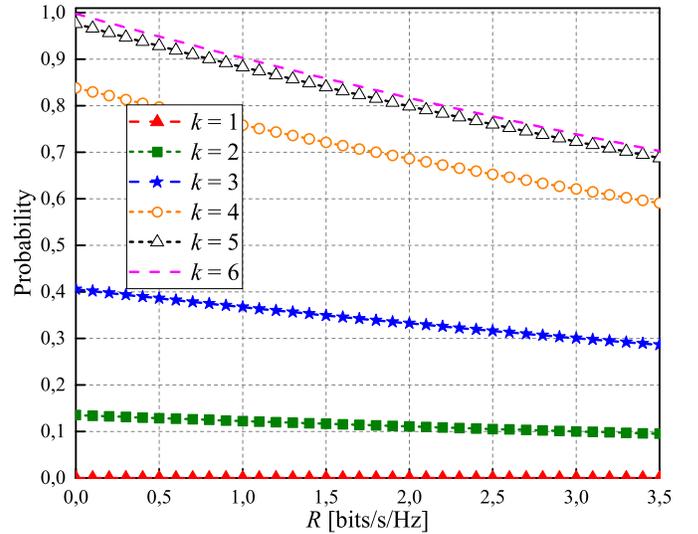}}
	\caption{Relay selection probability versus a rate threshold $R$ for the case when CSIT on the $L_i \to D$ link is not available at the relay and $\eta = 1$ for different number of time slots allocated for EH purposes, $k$, and for data transmission, $n$.}
	\label{fig:NCSI_prob}
\end{figure}
\begin{figure*}[!t]
	\centering
	\subfloat[Outage probability when $n=0$.]{
		\label{subfig:NCSI_prob1_n0}
		\includegraphics[width=3.5in]{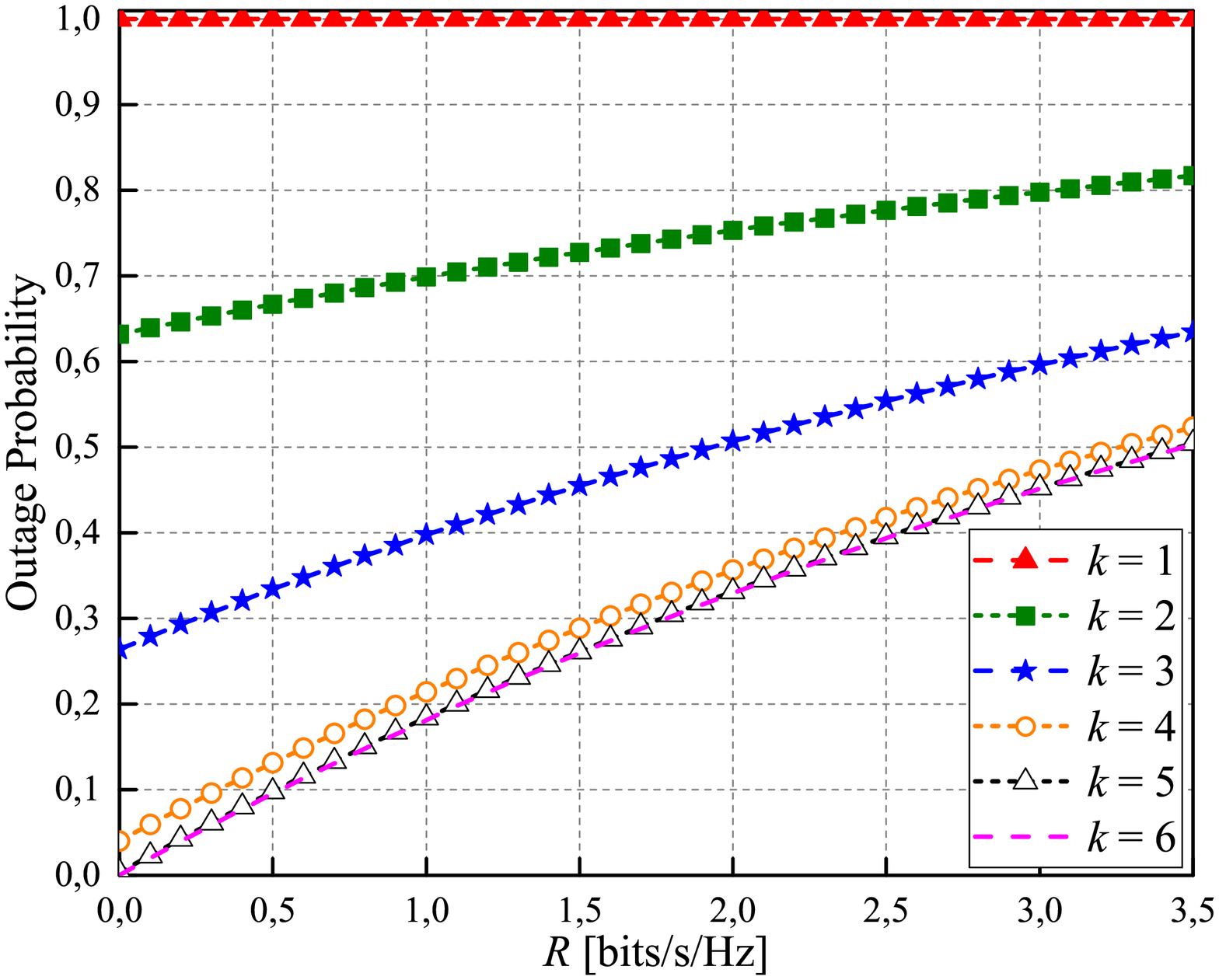}}
	\subfloat[Outage probability when $n=1$.]{
		\label{subfig:NCSI_prob1_n1}
		\includegraphics[width=3.5in]{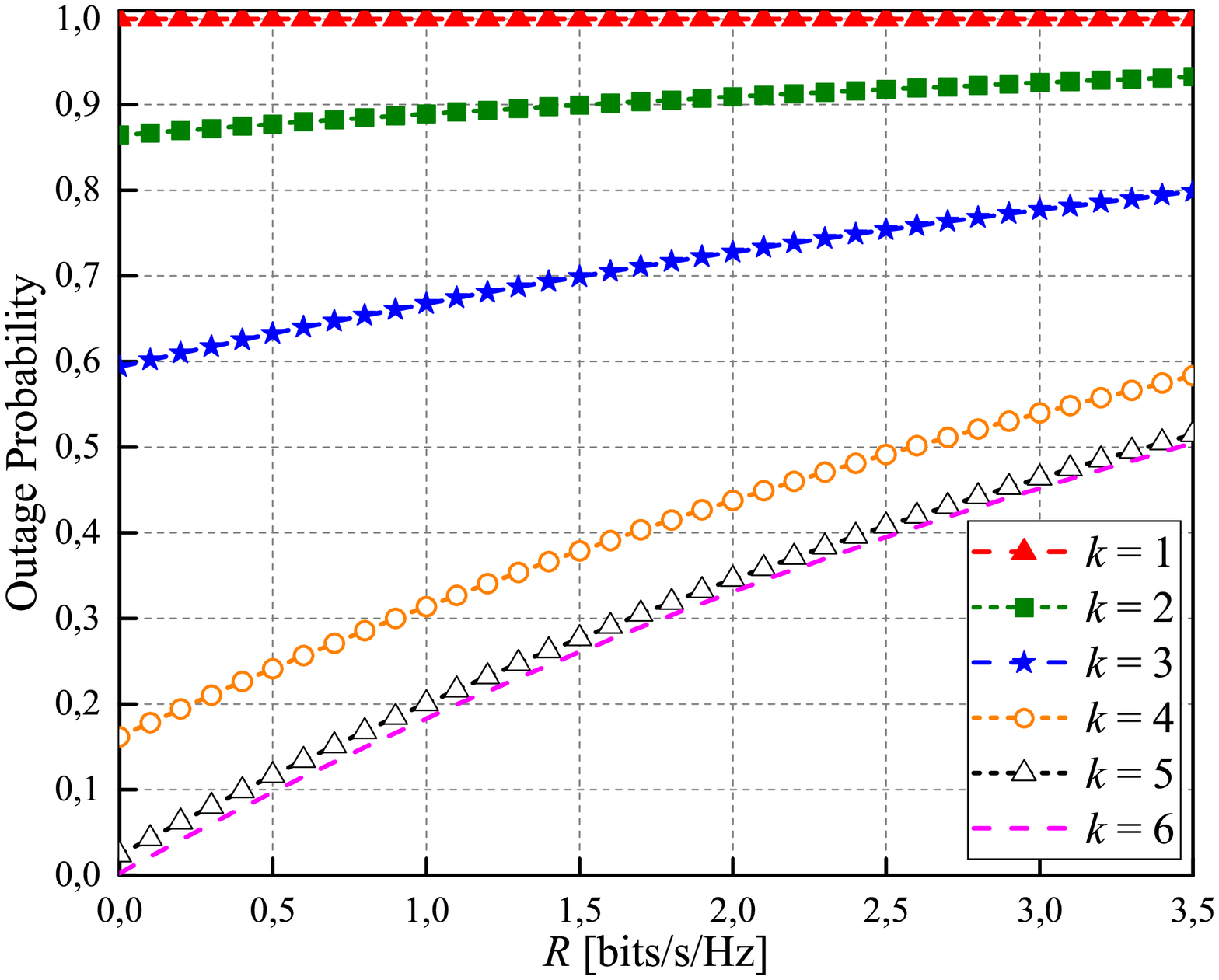}}
	\caption{Outage probability of a certain relay selected for data transmission versus a rate threshold $R$ for the case when CSIT on the $L_i \to D$ link is not available at the relay and $\eta = 1$ for different number of time slots allocated for EH purposes, $k$, and for data transmission, $n$.}
	\label{fig:NCSI_prob1}
\end{figure*}

In Fig. \ref{fig:SRS}, we compute the outage probability for the SRS-NCSI scheme for different number of available relay nodes. The CSIT is not available at the relay node on the $L_i \to D$ link and the selected relay node transmits with a fixed power 10 dBW. The outage probability decreases as the number of relays increases, as expected. However, when $N$ is sufficiently large, any further increase in $N$ does not benefit. For a small $N$, there is high probability that the selected relay stored energy $E_i^{\rm st}$ is not enough to transmit successfully on the $L_i \to D$ channel and 'power limitation' of the relay node contributes to the outage significantly. As $N$ increases, the outage performance improves. At $N=7$, the effect of power limitation vanishes completely and $N>7$ does not help to decrease outage. The system behaves like a grid powered system and the outage performance is given by (\ref{eqn:DF}).

In Fig. \ref{fig:NCSI_prob}, we demonstrate the probability that a certain relay will be selected for information transmission for the SRS-NCSI scheme (i.e., the CSIT is not available at the relay node on the $L_i \to D$ link and the selected relay node transmits with a fixed power $P_r = 10$ dBW) for different amounts of the harvested energy stored in the battery. The amount of energy is related to the number of time slots dedicated to EH purposes and to the number of time slots when this relay is acting as a communication node $n$. To investigate the number of time slots dedicated for EH, we assume $n = 0$, i.e., the relay has not been selected for information transmission yet. This assumption is reasonable for the case when the number of available relays is relatively large. The case $k=1$ means that the communication is just initiated and no power is available at the relays, and they operate in the EH mode only. It can be also observed that the probability of relay selection increases when the amount of the stored energy increases, as expected. However, further increase of the stored energy does not contribute to the probability of relay selection significantly, i.e., it starts saturating after $k \ge 5$. Next, for the case when $n = 1$ (a relay $i$ has been selected once), the selection probability severely deteriorates when the battery charge is low which, in turn, corresponds to the outage probability degradation shown in Fig. \ref{fig:NCSI_prob1}.

Next, we compare SRS-NCSI scheme with other similar available schemes. As a benchmark, we consider two commonly used schemes. In the first scheme, the relay is selected such that \cite{majid:ietsp16,Yaming,Krikidis_COML:2012},
\begin{equation}
i^* = \arg\max_i (E_i^{\rm st}-P_r)^+~\times I(R_{si}>R),
\label{eqn:bestenergy}
\end{equation}
with the notation $x^+=\max(x,0)$.
We denote it by SRS-NCSI-best-energy scheme, where the relay with the largest residual energy is selected for transmission.

\begin{figure}
\centering
  	\includegraphics[width=3.5in]{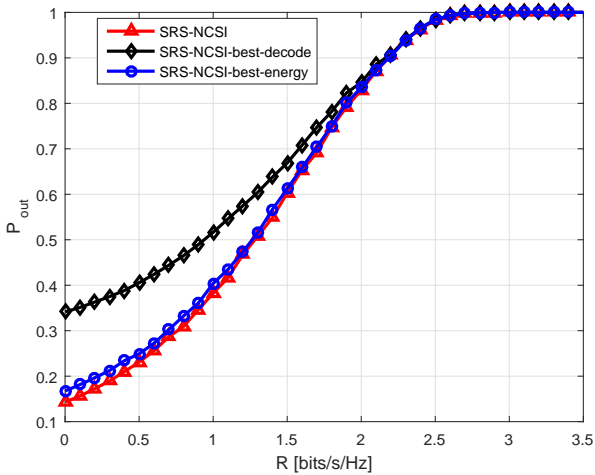}
   \caption{Comparison of the relay selection schemes for the case when $N=10$ and $\eta=0.1$, and the CSIT is not available at relay on the $L_i \to D$  link.}
	\label{fig:SRS_comp}
\end{figure}

The second scheme selects the relay which has the best chance of decoding on the $S\to L_i$ channel. The concept is similar to relay antenna selection scheme in \cite{Krikidis_TCOM2:2014}, where the antennas with large channel gains are selected for decoding. Thus,
\begin{equation}
i^* = \arg\max_i R_{si}\times I(R_{si}>R)~.
\label{eqn:bestdecoding}
\end{equation}
We denote this scheme by SRS-NCSI-best-decoding. Please note that the forwarding in all schemes is made only if $E_i^{\rm st}>P_r$, which saves transmit energy on unsuccessful transmission.

From Fig. \ref{fig:SRS_comp}, we see that SRS-NCSI outperforms the other schemes. The SRS-NCSI-best-energy performs better than SRS-NCSI-best-decoding because the major cause of outage is insufficient energy to forward data for the selected relay. The best channel selection on the $L_i \to D$ link is not optimal as decoding is already ensured by the condition $I(R_{si}>R)$.

\begin{figure*}[!t]
	\centering
	\subfloat[Outage probability when $n=0$.]{
		\label{subfig:NCSI_prob2_n0}
		\includegraphics[width=3.5in]{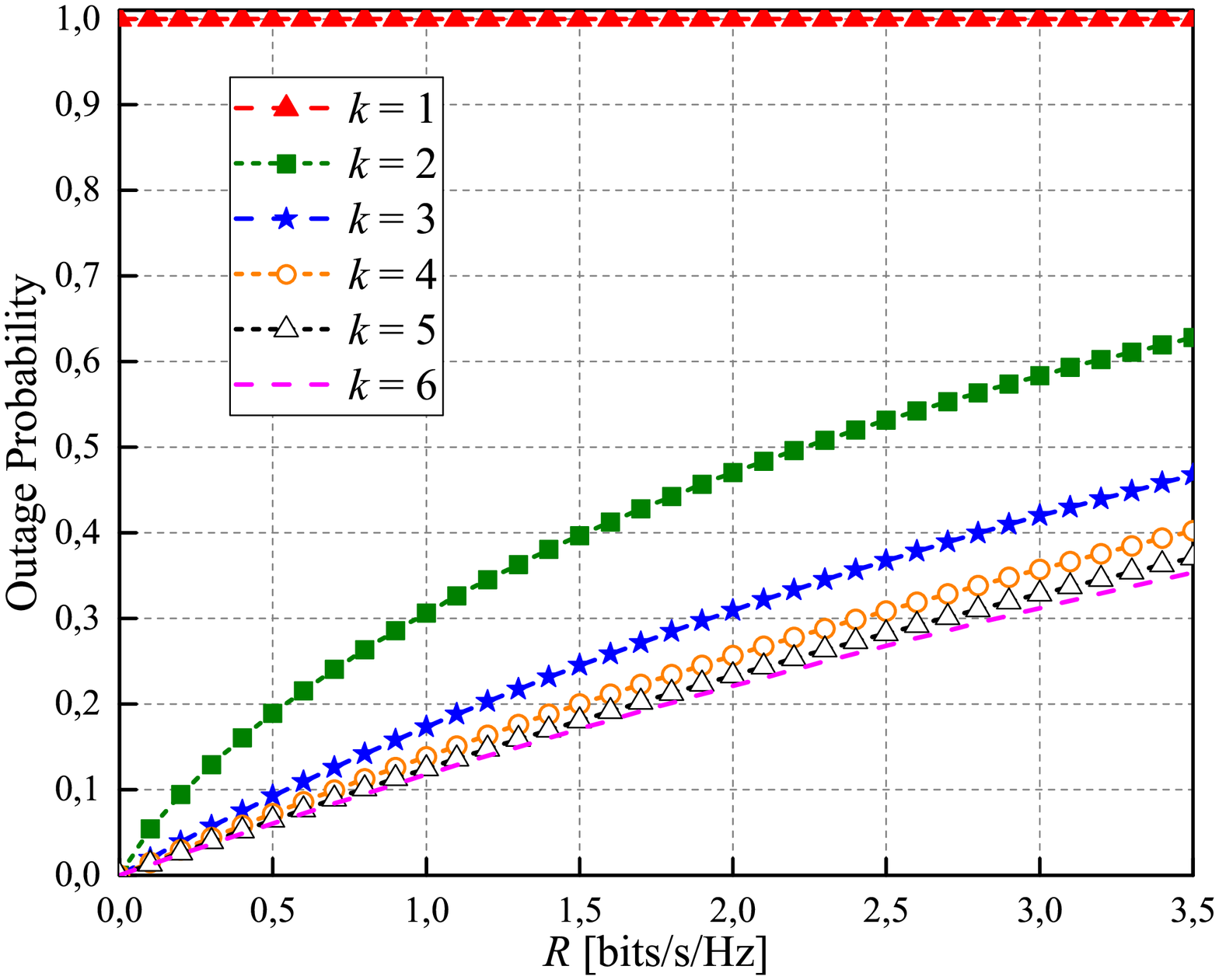}}
	\subfloat[Outage probability when $n=1$.]{
		\label{subfig:NCSI_prob2_n1}
		\includegraphics[width=3.5in]{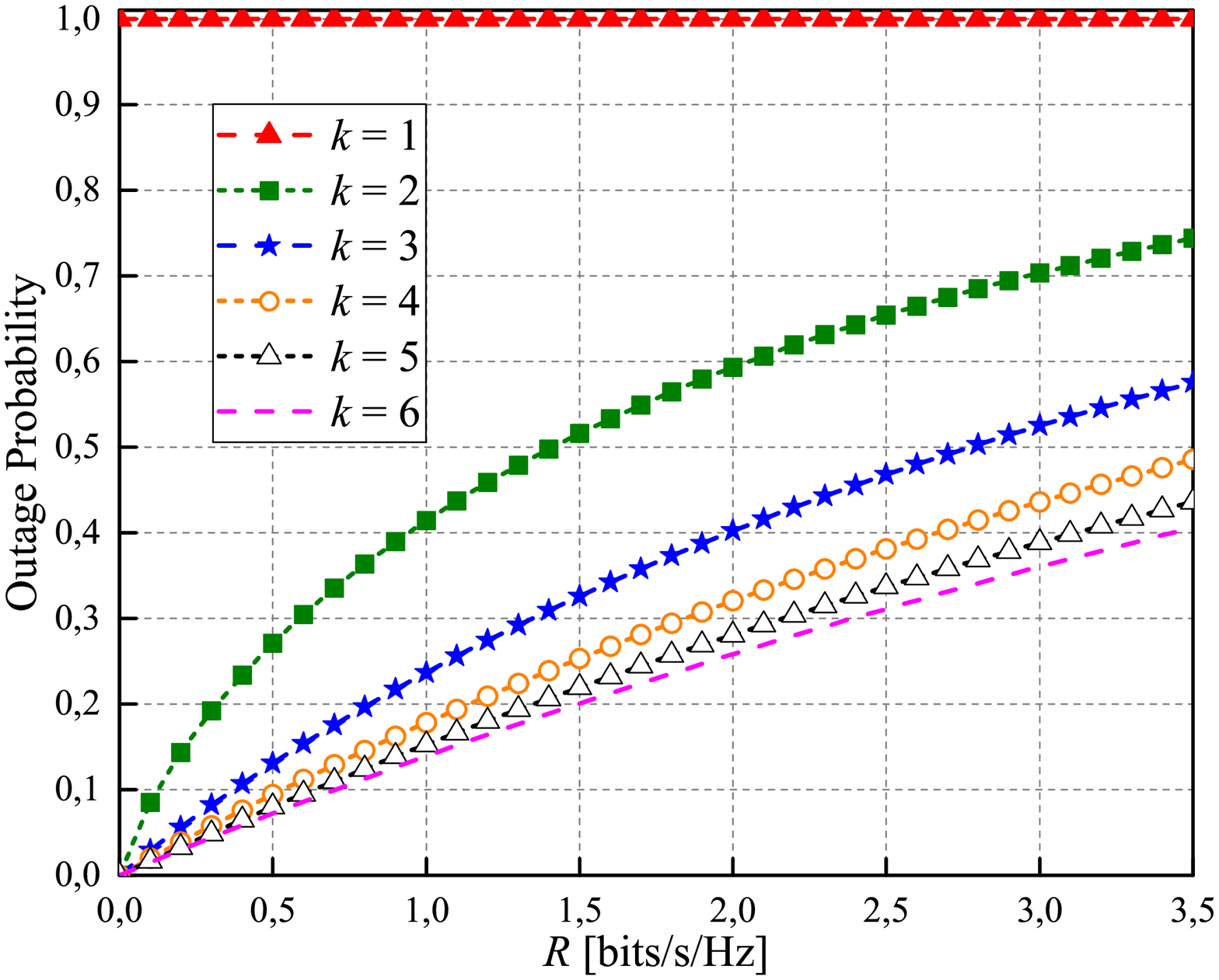}}
	\caption{Outage probability of a certain relay selected for data transmission versus a rate threshold $R$ for the case when CSIT on the $L_i \to D$ link is available at the relay and $\eta = 1$ for different number of time slots allocated for EH purposes, $k$, and for data transmission, $n$.}
	\label{fig:NCSI_prob2}
\end{figure*}

In Fig. \ref{fig:NCSI_prob2}, we compute the outage probability of the selected relay for the SRS scheme and various $k$ when CSIT on the $L_i \to D$ link is available at the relay node which transmits with $P_{id} = \frac{v \sigma^2}{|h_{id}|^2}$. The outage performance is shown versus the data rate threshold for $n = \{0,1\}$. In the case when CSIT is available, the outage probability outperforms that of the case when a relay is not aware of the channel properties of the $L_{i} \to D$ link. It is clear that the performance deteriorates when $n$ is increased, as expected. It is worthwhile pointing out that the result of the scenario with available CSIT always outperforms one when no CSIT on the $L_{i} \to D$ link is available at the relay (see Fig. \ref{fig:NCSI_prob1} for comparison).

\begin{figure}
\centering
  	\includegraphics[width=3.5in]{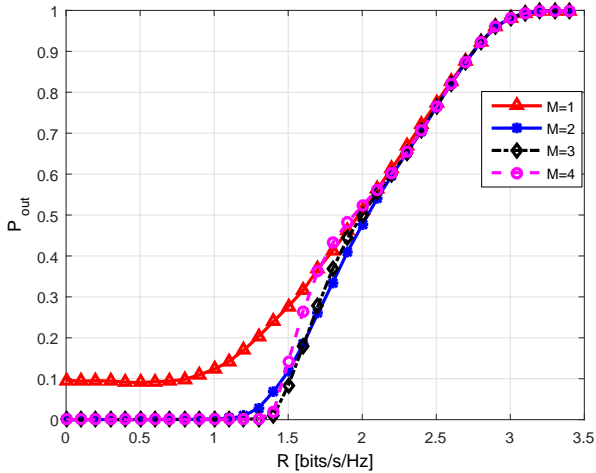}
   \caption{Outage probability for the MRS scheme for $N=10$ and $\eta=0.1$.}
	\label{fig:MRS}
\end{figure}

In Fig. \ref{fig:MRS}, we plot the outage probability for the MRS-ACSI scheme and compute the optimal value of parameter $M$ for a given $N$ and $\eta$. If $M$ is too small, multiuser diversity is not exploited effectively on the $L_i \to D$ link. On the contrary, if $M$ is too large, EH is not enough for the relays to store enough energy to avoid outage events. We observe that $M=3$ is optimal at small $R$, while $M=2$ is optimal at large $R$. This is attributed to the fact that large rate requirements require more relay nodes to harvest energy to have sufficient energy for successful transmissions to the destination.

\begin{figure}
\centering
  	\includegraphics[width=3.5in]{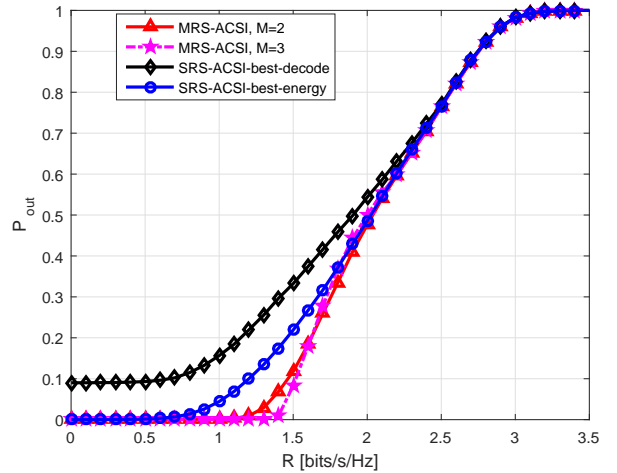}
   \caption{Comparison of the relay selection schemes for the case when $N=10$ and $\eta=0.1$, and the CSIT is available at the relay on the $L_i \to D$  link.}
	\label{fig:MRS_comp}
\end{figure}

\begin{figure}
	\centering
	\includegraphics[width=3.5in]{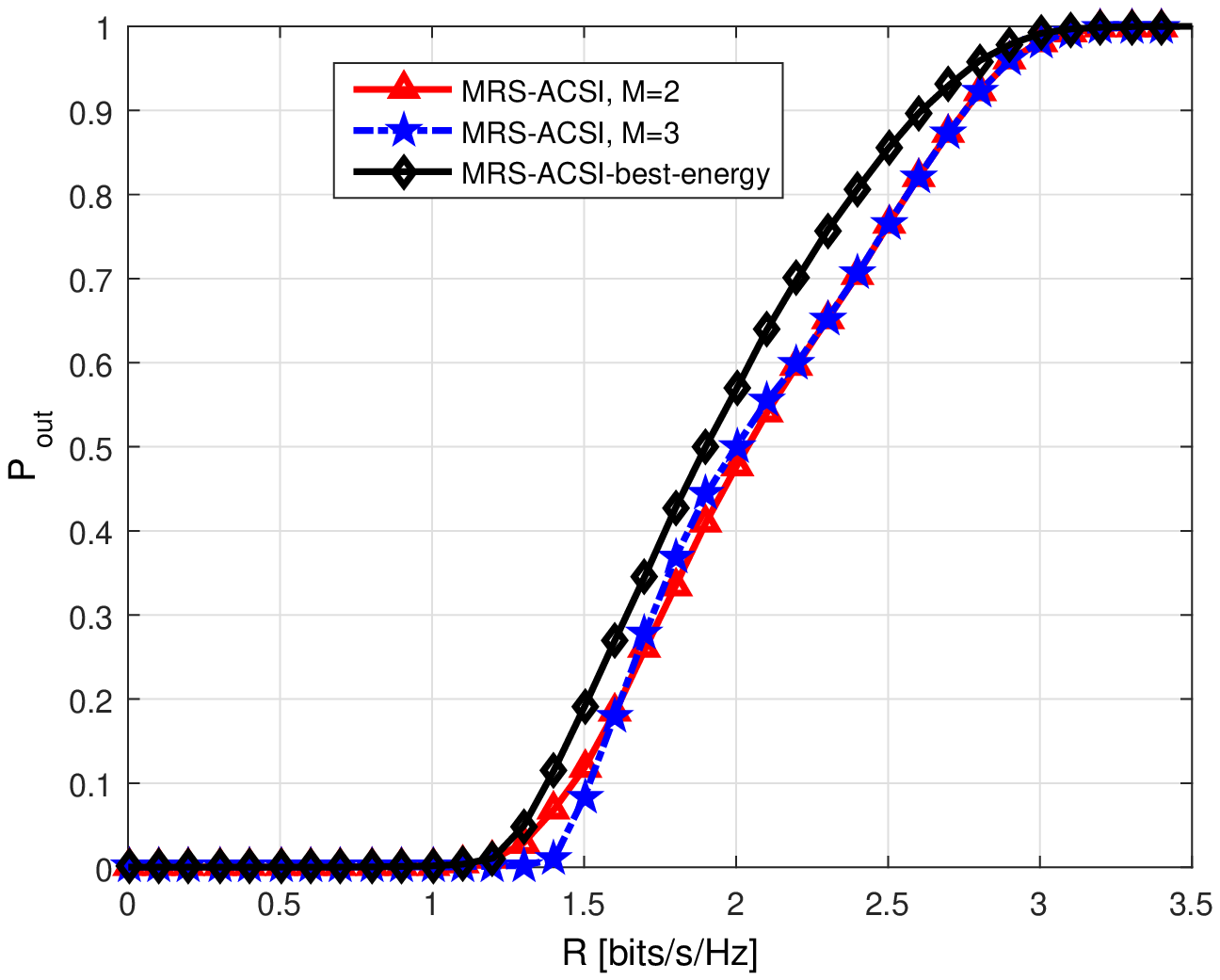}
	\caption{Comparison of the relay selection schemes for the case when $N=10$ and $\eta=0.1$, and the CSIT is available at the relay on the $L_i \to D$  link. Multiuser diversity is exploited using different relay selection metrics.}
	\label{fig:MRS_comp_group}
\end{figure}

In Fig. \ref{fig:MRS_comp}, we compare the performance of the MRS scheme with the two schemes mentioned above. We assume that the CSIT is available at relay on the $L_i \to D$ link and power $P_{id}$ is allocated for the selected relay by (\ref{eqn:power}). However, due to unavailability of $P_{id}(t+1)$ at time $t$, we eliminate transmit power term from (\ref{eqn:bestenergy}) and evaluate the metric,
\begin{equation}
i^* = \arg\max_i E_i^{\rm st}(t)~\times I(R_{si}>R)~.
\end{equation}
It is worth noting that the relay $L_{i^*}$ is not available for harvesting at time slot $t$ even if $E_i^{\rm st}<P_{id}(t+1)$ because $P_{id}$ can only be calculated at instant $t+1$ due to delayed CSIT on the $L_i \to D$ link. Fig. \ref{fig:MRS_comp} shows that power allocation due to available CSIT on the $L_i \to D$ channel at time $t+1$ improves the outage performance for the SRS-NCSI-best-energy and SRS-NCSI-best-decoding schemes as compared to their respective performance in Fig. \ref{fig:SRS_comp}, but MRS-ACSI outperforms both schemes comfortably due to inherent multiuser diversity exploitation.

To demonstrate the effect of relay selection metrics in (\ref{eqn:csi_gamma2}) and (\ref{eqn:MRS_Ph2}), we compare MRS-ACSI scheme with a similar 2-phase relay selection scheme proposed in \cite{majid:ietsp16}. Like MRS-ACSI, $M$ relays with the largest stored energies are selected in first phase. In the second phase, a relay $i$ out of $M$ relays is selected such that,
\begin{equation}
i^* = \arg\max_{i\in \Gamma} \Big(E_{\rm st}^i(t+1)-P_{id}(t+1)T\Big).
\end{equation}
This scheme is denoted by MRS-ACSI-best-energy. For the MRS-ACSI-best-energy scheme, $M^*=5$ for the parameters $\eta=0.1, N=10$ \cite{majid:ietsp16}. The results in Fig. \ref{fig:MRS_comp_group} reveal that our scheme performs better that the MRS-ACSI-best-energy scheme in spite of the fact that both schemes are exploiting multiuser diversity. From the numerical evaluation in Fig. \ref{fig:MRS_comp} and Fig. \ref{fig:MRS_comp_group}, we conclude that based on exploitation of multiuser diversity and careful design metrics for both phases as explained in Section \ref{sect:scheme_ACSI}, our proposed MRS-ACSI scheme performs better that the other schemes available in literature.

\section{Conclusions}
\label{sect:conclusion}
We propose novel relay selection schemes for the WPCNs and discuss the scenarios where the CSI is available at relay on the $S\to L_i$ link.
At the time of selection of relay at the $S\to L_i$ link, the unavailability of any information on CSI at the $L_i\to D$ link is a very practical scenario as compared to ideal scenarios where CSI availability for the fading channels is assumed throughout at the time of relay selection at the source.
Conditioned on the availability of the CSIT on the $L_i \to D$ link at time slot $t+1$, we propose two heuristic relay selection schemes to minimize system outage as the problem is not tractable due to unavailability relay power at the time of relay selection. To provide insight, we provide closed-form analytical expressions of the outage probability for both scenarios.
Next, we evaluate the performance of the proposed schemes numerically and compare it with the commonly used relay selection schemes. When the CSI is available at the relay on both $S\to L_i$ and $L_i \to D$ links, consideration of mutually independent i.i.d. channels on two hops makes the half duplex relay selection problem challenging. A two--phase relay selection scheme in conjunction with our proposed relay selection metrics is proposed to exploit the multiuser diversity effectively. The numerical evaluation shows that our proposed scheme outperforms the other schemes from the literature comfortably when power allocation is applied on the $L_i \to D$ link.

\section*{Acknowledgement}
The authors are thankful to Ioannis Krikidis for useful discussions that helped immensely to improve the quality of the paper.
\begin{appendices}

\section{Outage for the SRS scheme without CSIT}
\label{sect:Appen1}
Since the relay node selected for information transmission does not simultaneously harvest energy, the probability that a certain relay $i$ will be selected for information delivery at time moment $k$ can be expressed as
\begin{align}
\label{A_i}
A_i(k) = \text{Pr}\left(R_{si}(t+k) > R,~P_{r_{i}}(t+k-1)>P_r\right),
\end{align}
where $P_{r_{i}}(t+k-1)$ denotes the overall amount of stored harvested power at a relay $i$ after previous $k-1$ time slots and can be expressed using \eqref{4}, \eqref{energy_NCSI1} and \eqref{energy_NCSI2} as
\begin{align}
	P_{r_{i}}(t+k-1) &= \eta P_s S_{k-1} - n P_r,
\end{align}
where $S_{k-1} = \sum_{l=1}^{k-1}X_i$ and $X\sim Exp(\lambda)$, where $\lambda$ is the mean of the exponential RV $|h|^2$. $0\le n \le \left\lfloor \frac{k}{2}\right\rfloor$ denotes the number of time slots when this relay was chosen for the relay-to-destination transmission, and its maximum equals $ n = \left\lfloor \frac{k}{2}\right\rfloor$ when the relay is chosen for information transmission every two time slots\footnote{This scenario is applicable when a number of relays is low and will result in the high system outage since the energy harvested within one or two time slots, due to the channel randomness, is less likely to be sufficient to support information transmission with a fixed transmit power $P_r$, i.e., the higher is  the $P_r$ required the more severe outage occur.}.

Considering i.i.d. Rayleigh fading channels, the PDF of $S_l$ can be given for different $k$ as
\begin{align}
f_{S_{l}}(z) = \begin{cases}
0, \hfill \text{for}~k=1,\\
\lambda\exp\left(-\lambda z\right), z>0, \hfill \text{for}~k=2,\\
\lambda^2 z \exp\left(-\lambda z\right), z>0, \hfill \text{for}~k=3,\\
\frac{\lambda \left(\lambda z\right)^{l-1} \exp\left(-\lambda z\right)}{\Gamma(l)}, z>0, \hfill \text{for}~k>3.
\end{cases}
\end{align}
Therefore, due to the independence of the involved RVs, the relay selection probability can be further written as
\begin{align}
A_i(k) &= \text{Pr}\left(|h_{si}(t+k)|^2 > \frac{v\sigma^2}{P_s},
S_{k-1} > \frac{(n+1)P_r }{\eta P_s}
\right)\nonumber\\
&= \exp\left( - \frac{\lambda v\sigma^2}{P_s}\right) \left(1 - F_{S_{l}}\left(\underset{U}{\underbrace{\frac{(n+1)P_r }{\eta P_s}}}\right)\right),
\end{align}
where $v = 2^{2 R} - 1$ denotes the SNR value associated with the rate threshold $R$. $F_{S_{l}}\left(U\right)$ is the cumulative distribution function (CDF) of RV $S_{l}$ given by
\begin{align}
F_{S_{l}}(U) = \begin{cases}
1, \hfill \text{for}~k=1,\\
1 - \exp\left(-\lambda U\right), \hfill \text{for}~k=2,\\
\frac{1 - \left(\lambda U + 1 \right)\exp\left(-\lambda U\right)}{\lambda}, \hfill \text{for}~k=3,\\
\frac{\gamma_{inc}\left(l,\lambda U\right)}{\Gamma(l)}, \hfill \text{for}~k>3,
\end{cases}
\end{align}
where $\gamma_{inc}(s,x) = \int_{0}^{x} t^{s-1} \exp(-t){\rm d}t$ denotes the lower incomplete Gamma function \cite{Gradst}.

With respect to \eqref{eqn:SRS_selection}, the outage probability $P_{\rm out}$ given by \eqref{eqn:outage_EH} can be rewritten using \eqref{A_i} as
\begin{align}
P_{{\rm out},i}(k) &= \text{Pr}\left(\min\left(R_{si}(t+k),R_{id}(t+k)\right) < R\right) \nonumber\\
&= \text{Pr}\left(R_{id}(t+k) < R, R_{si}(t+k) > R, \right. \nonumber\\
& \hspace{0.9cm} \left. P_{r_i}(t+k - 1) > P_r\right) \nonumber\\
&= 1 - \text{Pr}\left(R_{id}(t+k) > R, R_{si}(t+k) > R, \right. \nonumber\\
& \hspace{1.5cm} \left. P_{r_i}(t+k - 1) > P_r\right) \nonumber\\
&= 1 - \exp\left( -\frac{\lambda v \sigma^2}{P_r} \right) A_i(k).
\end{align}

These derivations prove the results on outage performance presented in Proposition 1.

\section{Outage for the SRS scheme with CSIT}
\label{sect:Appen2}
In the case when CSIT on the $L_i \to D$ link is available at the relay nodes, the selected relay can efficiently allocate the power as $P_{id} = \frac{v\sigma^2}{|h_{id}|^2}$. Therefore, the outage probability of a relay $i$ can evaluated as
{\allowdisplaybreaks
\begin{align}
P_{\rm out} &= \text{Pr}\left(\min\left(\frac{|h_{si}|^2 P_s}{\sigma^2}, \frac{|h_{id}|^2 P_{id} }{\sigma^2} \right) < v\right) \nonumber\\
&= \text{Pr}\left(
\frac{|h_{id}|^2 P_{id} }{\sigma^2} < v,
|h_{si}|^2 > \frac{v \sigma^2}{P_s}, P_{r_i}> P_{id}\right) \nonumber\\
&= 1 - \text{Pr}\left(
\frac{|h_{id}|^2 P_{id} }{\sigma^2} \ge v,
|h_{si}|^2 > \frac{v \sigma^2}{P_s}, P_{r_i}> P_{id}\right) \nonumber\\
&= 1 - \text{Pr}\left(
v \ge v,
|h_{si}|^2 > \frac{v \sigma^2}{P_s}, |h_{id}|^2 > \frac{(n+1)v\sigma^2}{\eta P_{s} S_{l}}
\right) \nonumber\\
&= 1 - \text{Pr}\left(
|h_{si}|^2 > \frac{v \sigma^2}{P_s}\right) \text{Pr}\left( |h_{id}|^2 > \frac{(n+1)v\sigma^2}{\eta P_{s} S_{l}}
\right) \nonumber\\
&= 1 - \exp\left(-\frac{\lambda v \sigma^2}{P_s} \right) \underset{B(k)}{\underbrace{\text{Pr}\left( |h_{id}|^2 > \frac{Q}{S_{l}}\right)}},
\end{align}}where $Q = \frac{(n+1)v \sigma^2}{\eta P_s}$ and $B(k)$ is the complementary CDF defined as in \eqref{Bk}, shown at the top of the next page, where $K_1$, $G^{m,n}_{p,q}(\cdot)$ and $H^{m,n}_{p,q}(\cdot)$ denote the modified Bessel function of the second kind of order 1, the Meijer G-function \cite[(8.4.3.1)]{Prudnikov} and the Fox's H-function \cite[(1.2)]{mathai2009h}, \cite{8395375}, respectively.

This proves the result in Proposition 2.

\begin{figure*}[!t]
\begin{align}
\label{Bk}
B(k) &= \int_{0}^{\infty}f_{S_l}(z) \left[1 - F_{|h_{id}|^2}\left( \frac{Q}{z} \right)\right] \text{d}z = \int_{0}^{\infty}f_{S_l}(z) \exp\left(- \frac{\lambda Q }{z} \right) \text{d}z \nonumber\\
&= \begin{cases}
0, \hfill \text{for}~k=1,\\
2
\sqrt{\lambda Q}
K_1\left( 2\sqrt{\lambda Q} \right), \hfill \text{for}~k=2, \\
\left(\lambda Q\right)^2 G^{2,0}_{0,2}\left(\lambda Q \left|
\begin{array}{c}
-\\
0,~-2
\end{array}
\right.\right), \hfill \text{for}~k=3,\\
\frac{1}{\Gamma(k-1)}\left(\lambda Q\right)^{k-1} H^{2,0}_{0,2}\left(\lambda Q \left|
\begin{array}{c}
-\\
(0,1),~(-(k-1),1)
\end{array}
\right.\right),~\text{for}~k> 3.\\
\end{cases}
\end{align}
\hrulefill
\end{figure*}

\end{appendices}
\renewcommand{\baselinestretch}{1}
	
\renewcommand{\bibfont}{\scriptsize}
\bibliographystyle{IEEEtran}
\bibliography{bibliography}
\end{document}